\documentclass[10pt,reqno]{amsart}

\usepackage{latexsym}

\usepackage{amssymb, amscd, amsthm, multirow}
\usepackage[all]{xy}
\usepackage{verbatim}
\usepackage{bbm}
\usepackage{bm}
\usepackage{mathrsfs}
\usepackage{color}
\usepackage{booktabs}
\usepackage{threeparttable}

\usepackage[numbers,sort&compress]{natbib}

\newtheorem{theorem}{{Theorem}}
\newtheorem{lemma}{{Lemma}}

\newtheorem{conj}{{Conjecture}}
\newtheorem{proposition}{{Proposition}}
\newtheorem{definition}{{Definition}}

\theoremstyle{remark}
\newtheorem{remark}{{Remark}}

\theoremstyle{definition}
\newtheorem{example}{Example}

\newcommand{\F}{\mathbb{F}}

\newcommand{\Z}{\mathbb{Z}}

\renewcommand{\mod}{\,\mathrm{mod}\,}
\newcommand{\Tr}{\mathrm{Tr}}

\newcommand{\bs}{\mathbf{s}}

\newcommand{\wt}{\mathrm{wt}}

\newcommand{\C}{\mathcal{C}}
\newcommand{\cD}{\mathcal{D}}
\newcommand{\cS}{\mathcal{S}}

\usepackage[shortlabels, inline]{enumitem}
\SetEnumerateShortLabel{a}{\textup{(\alph*)}}
\SetEnumerateShortLabel{A}{\textup{(\Alph*)}}
\SetEnumerateShortLabel{1}{\textup{(\arabic*)}}
\SetEnumerateShortLabel{i}{\textup{(\roman*)}}
\SetEnumerateShortLabel{I}{\textup{(\Roman*)}}

\begin{document}
\title[codes from sequences]{Binary $[n,(n\pm1)/2]$ cyclic codes with good minimum distances from sequences}

\author{Xianhong Xie$^1$, Yaxin Zhao$^1$, Zhonghua Sun$^2$, Xiaobo Zhou$^1$}
\address{$^1$School of Information and Artificial Intelligence, Anhui Agricultural University, Hefei, Anhui 230036, China}
\address{$^2$School of Mathematics, Hefei University of Technology,
	Hefei, Anhui 230601, China}

%\address{$^3$Hefei National Laboratory, University of Science and Technology of China, Hefei 230088, China}
\email{xianhxie@ahau.edu.cn}
\email{zyx20001130@163.com}
%\email{yiouyang@ustc.edu.cn}
\email{sunzhonghuas@163.com}
\email{zxb@ahau.edu.cn}

\thanks{Partially supported by  National Natural Science Foundation of  China under Grant 62371004 and 62002093.}

\begin{abstract}
Recently, binary cyclic codes with parameters $[n,(n\pm1)/2,\geq \sqrt{n}]$ have been a hot topic since their minimum distances have a square-root bound. In this paper, we construct four classes of binary cyclic codes $\C_{\cS,0}$, $\C_{\cS,1}$ and $\C_{\cD,0}$, $\C_{\cD,1}$ by using two families of sequences, and obtain some codes with parameters  $[n,(n\pm1)/2,\geq \sqrt{n}]$.  For $m\equiv2\pmod4$, the code $\C_{\cS,0}$ has parameters $[2^m-1,2^{m-1},\geq2^{\frac{m}{2}}+2]$, and the code $\C_{\cD,0}$ has parameters $[2^m-1,2^{m-1},\geq2^{\frac{m}{2}}+2]$ if $h=1$ and $[2^m-1,2^{m-1},\geq2^{\frac{m}{2}}]$ if $h=2$.

\smallskip

\noindent\textbf{Keywords} Cyclic code, Minimum distance, BCH bound, Si-Ding sequence, Ding-Zhou sequence.
\end{abstract}
\maketitle

\section{Introduction}
Set $n=2^m-1\geq 1$.
%A sequence $\bs^{\infty}=\{s_{0},s_{1},s_{2},\ldots\}$ has period $n$ if $s_{i}=s_{i+n}$ for all $i\geq0$, where $s_i\in \F_{2^m}$.  For $\bs^\infty$  a sequence of period $n$, the characteristic power series/polynomial of $\bs^{\infty}$ and $\bs^{n}=\{s_{0},s_{1},\ldots,s_{n-1}\}$ are defined respectively as $c^{\infty}(x)=s_{0}+s_{1}x +\cdots$ and $c^{n}(x)=s_{0}+s_{1}x+\cdots+s_{n-1}x^{n-1}$, the minimal polynomial \cite{2} of $\bs^\infty$ is \begin{equation*} m(x)=(x^{n}-1)/\gcd(c^{n}(x),x^{n}-1). \end{equation*}
%An $[n,k,d]$ binary cyclic code $\C$ can be viewed as an ideal of the ring $\F_{2^m}[x]/\langle x^n-1\rangle$, there exists a polynomial $g(x)\in\F_{2^m}[x]/\langle x^n-1\rangle$ with smallest degree such that $\C=\langle g(x)\rangle$. The polynomial $h(x)=(x^n-1)/g(x)$ is called the check polynomial of $\C$.% and  its reciprocal can generate the dual code $\C^\bot$ of $\C$.
 Binary cyclic codes have been studied for decades and a lot of progress has been made since they  have wide applications in storage and communication systems. One method of constructing cyclic codes over $\F_{2^m}$ is to take \[m(x)=\frac{x^n-1}{\gcd(c^n(x),x^n-1)}\] as the generator polynomial, where $c^{n}(x)=s_{0}+s_{1}x+\cdots+s_{n-1}x^{n-1}\in\F_{2^m}[x]$ and  $\bs=\{s_{0},s_{1},s_{2},\ldots\}$ is a binary sequence over $\F_{2^m}$ with period $n$. The polynomial $m(x)$ is called the minimal polynomial of $\bs$ and the degree of $m(x)$ is referred to its linear complexity.  Throughout this paper, the cyclic code $\C_{\bs}$ with the generator polynomial $m(x)$ is called the code defined by the sequence $\bs$.
 % The degree of the polynomial $m(x)$ is referred to as the linear complexity of $\bs^\infty$, denoted by $\LC(\bs^\infty)$. Then the dimension of $\C_{\bs}$ is $n-\LC(\bs^\infty)$. The linear complexity of a sequence is an important criteria of its quality. As we all know, sequences with high linear complexity (such that $\LC(\bs^\infty)>n/2$)  have important applications in cryptography. Ding used the two-prime sequences and the cyclotomic sequences of order four to construct cyclic codes from this way in \cite{Din1} and \cite{dinc} respectively.

It is well known that for any sequence $\bs=(s_t)_{t=0}^{n-1}$ over $\F_{2^m}$ of period $n$, the component $s_t$ of $\bs$ can be written as
\begin{equation}\label{se}s_t=\sum_{i=0}^{n-1}a_i\alpha^{it},\quad a_i\in\F_{2^m},\ 0\leq t\leq n-1.\end{equation}
From \cite{Ant}, we know that the minimal polynomial $m(x)$ of the sequence of ~\eqref{se} is \[\prod_{i\in I}(1-\alpha^i x),\] where $I=\{i:a_i\neq0\}$. Note that some of the sequences with good properties, but not all such sequences, can produce optimal cyclic codes. Thus, an interesting problem is proposed: how to select the suitable sequence $\bs$ such that $\C_{\bs}$ has good parameters.
 %An impressive work in the construction of cyclic codes with this approach has been made. If the set $I$ only takes a cyclotomic coset, then the code $\C_{\bs}$ is called a irreducible cyclic code. If the set $I$ takes two or more cyclotomic cosets, the corresponding codes are called reducible cyclic codes.Ding in \cite{Din2} also proposed another way to construct cyclic codes via defining the following sequence
%\begin{equation}\label{s2}s_t=\Tr_{2^m/2}(f(1+\alpha^t)-f(\alpha^t)),\quad t\geq0.\end{equation}
%They pointed out that most of the functions constructed sequences in \eqref{s1} can be employed to construct cyclic codes from the sequences defined by ~\eqref{s2}.

Ding \cite{Din1,Din5,dinc} firstly constructed some cyclic codes with good parameters by using the two-prime sequences and the cyclotomic sequences of order four. Subsequently, Ding\ \cite{Din2} defined a new sequence $\bs=(s_0,s_1,\ldots,s_{n-1})$ by using the polynomial
$f(x)$ over $\F_{2^m}$, where  \begin{equation}\label{s1}s_t=\Tr_{2^m/2}(f(1+\alpha^t)),\quad t\geq0.\end{equation}

 By selecting suitable monomials and trinomials $f(x)$ over $\F_{2^m}$, some binary cyclic codes with good properties were presented in \cite{Din2,Din3}, and the lower bounds on the minimum distance were developed. For more results of sequence construction, please refer to the survey paper \cite{Din4}. Recently, binary cyclic codes with parameters $[n,(n\pm1)/2,d\geq\sqrt{n}]$ have been widely studied, but only a few cases are known (see \cite{Din5,Xiong1,Xiong2,Liu,Sun}). The objective of this paper is to study the parameters of the codes $\C_{\bs}$ defined by $f(x)=x^{2^m-2}$ and $x+x^{2^m-2}+x^{2^h-1}$, and then obtain some codes with parameters $[n,(n\pm1)/2,d\geq\sqrt{n}]$.

For $f(x)=x^{2^m-2}$, the corresponding sequence is denoted as $\cS$. Ding \cite{Din2} proved that the minimum distance $d(\C_\cS)$ satisfies $d(\C_\cS)^2\geq n$ if $2\nmid m$, and $d(\C_\cS)$  may not have $d(\C_\cS)^2\geq n$ if $2\mid m$. The first work of this paper is to study the parameters of $\C_{\cS}$ when $m$ is even. We first provide a new method to prove that the linear complexity of $\cS$ is $2^{m-1}$, and then give the dimension of $\C_\cS$. Secondly, we redefine the generator polynomial of $\C_{\cS}$, showing that the new codes $\C_{\cS,0}$ and $\C_{\cS,1}$ are subcodes of binary cyclic codes defined by Tang et al.\ \cite{Tang1}. Finally, we give the lower bounds on $d(\C_{\cS,0})$ and $d(\C_{\cS,1})$, and obtain some codes with parameters $[2^m-1,2^{m-1}-2,\geq 2^{\frac{m}{2}}+2]$ and $[2^m-1,2^{m-1},\geq 2^{\frac{m}{2}}+2]$.

For $f(x)=x+x^{2^m-2}+x^{2^h-1}$, where $0\leq h\leq \lceil\frac{m}{2}\rceil$, the corresponding sequence is denoted as $\cD$. Ding \cite{Din3}\ showed that the minimum distance $d(\C_{\cD})\geq8$ if $h=0$ and $2\nmid m$ , and $\geq3$ if $h=0$ and $2\mid m$. However, for $h>0$, the lower bounds on  $d(\C_{\cD})$ can not be given. The second work of this paper is to give new lower bounds on the minimum distance of the code $\C_{\cD}$, and showing that there exist some codes with parameters $[2^m-1,2^{m-1}-2,\geq 2^{\frac{m}{2}}+4]$ and $[2^m-1,2^{m-1},\geq 2^{\frac{m}{2}}]$.

\section{Preliminaries}
In this section, we give some basics about cyclotomic coset. From now on, we always assume $q=2^m$ and $n=q-1$.
\subsection{Notations} Throughout this paper we shall adopt the following notations.
\begin{itemize}
\item For any $i$ with $1\leq i\leq 2^m-1$, its $2$-adic expansion is $i=\sum\limits_{j=0}^{m-1}i_j2^j$, $i_j\in\{0,1\}$, define $\wt_2(i)=\sum\limits_{j=0}^{m-1}i_j$.

\item For $i\in\Z$, let $v_2(i)$ be the $2$-adic valuation of $i$, and $i^c:=\frac{i}{v_2(i)}$.

\item For $i\in\Z_n=\{0,1,\ldots,n-1\}$, let \[C_i=\{i2^{j}\bmod n,\ 0\leq j\leq l_i-1\},\] where $l_i$ is the least integer such that $i\equiv i2^{l_i}\bmod n$. Let $\Gamma$ denote the set of all coset leaders.

\item $\alpha$ is a generator of $\F_{2^m}^*$.%, $m_{\alpha}(x)$ denote the minimal polynomial of $\alpha$ over $\F_2$.

\item The set $T:=\{i:m(\alpha^i)=0,0\leq i\leq n-1\}$ is called the defining set of $\C$ with respect to the $n$-th primitive root of unity $\alpha$, where $m(x)\in\F_q[x]$ is a generator polynomial of $\C$.

\item For $m>0$, $\Tr_{2^m/2}(x)=\sum\limits_{t=0}^{m-1}x^{2^t}$ is the trace map from $\F_{2^m}$ to $\F_2$.

\item According to the Database we mean the collection of the tables of best linear codes known maintained by Markus Grassl at \emph{http://www.codetables.de/}.
\end{itemize}
\subsection{Cyclotomic coset}
For any integer $t$ with $1\leq t\leq m$, set
\begin{equation}\label{c}
  \epsilon_a^{(t)}=\begin{cases}
   1, & \mbox{if } a=2^t-1 \\
\bigl\lceil\log_2(\frac{2^t-1}{a})\bigr\rceil, & \mbox{if } 1\leq a<2^t-1.
    \end{cases}
\end{equation}The following lemma can be directly obtained from the definition of $\epsilon_a^{(t)}$.
\begin{lemma}\label{lem:ep}
  Set $1\leq a\leq 2^t-2$. Then
  \begin{itemize}
    \item $\epsilon_a^{(t+1)}=\epsilon_a^{(t)}+1$;
    \item $\epsilon_{2^k a}^{(t)}=\epsilon_a^{(t)}-k$ for any $0\leq k\leq \epsilon_a^{(t)}-1$.
  \end{itemize}
\end{lemma}

 Set $C_a^{(t)}:=\{2^ia:0\leq i\leq \epsilon_a^{(t)}-1\}$. By Lemma~\ref{lem:ep}, we have \[C_{2^ka}^{(t)}\subseteq C_a^{(t)}\ \text{for any}\ 0\leq k\leq \epsilon_a^{(t)}-1.\] Thus we can assume $a$ is odd and the following result can be obtained
\[\bigcup_{1\leq 2j+1< 2^t-1}C_{2j+1}^{(t)}=\Z_{2^t-1}\setminus\{0\}=\{1,2,\ldots,2^t-2\}.\]

For any odd integer $a$ with $1\leq a\leq 2^t-2$, there must exist a smallest integer $l_a$ such that $a2^{l_a}\equiv a\pmod{2^t-1}$. Take $s_a:=1+2^{l_a}+2^{2l_a}+\cdots+2^{t-l_a}$. Then $a$ can be written as
\begin{equation}\label{aq}a=s_a+i_12s_a+
\cdots+i_{l_a-1}2^{l_a-1}s_a,\ i_j\in\{0,1\}.\end{equation}

\begin{proposition}\label{lem2}Set $N:=|\{j:i_j=1,1\leq j\leq l_a-1\}|$ satisfying $N<l_a-1$. Then \[C_a:=\{a2^i:\ 0\leq i\leq l_a-1\}=\bigcup_{i=1}^N C_{2^{l_a-j_i}a}^{(t)}\bigcup C_{a}^{(t)},\]
i.e., the number of even integers in $C_a$ and $C_{a^c}$ are equal respectively to $l_a-N-1$ and $N+1$.

In particular, if $l_a=t$, then the number of even integers are $t-\wt_2(a)$
and $\wt_2(a)$, respectively.
\end{proposition}
\begin{proof}Let $i_{j_1}=i_{j_2}=\cdots=i_{j_N}=1$ satisfying $1\leq j_1<j_2<\cdots<j_N\leq l_a-1$. Set $j_0=0$ and $j_{N+1}=l_a$. For two integers $j_l,j_k$ with $j_1\leq j_l<j_k<\cdots\leq j_{N+1}$, we have
\begin{align*}&C_{2^{l_a-j_l}a}^{(t)}=\{2^{l_a-j_l}a,2^{l_a-j_l+1}a,\ldots,2^{l_a-j_{l-1}-1}a\}\ \text{and}\\ &C_{2^{l_a-j_k}a}^{(t)}=\{2^{l_a-j_k}a,2^{l_a-j_k+1}a,\ldots,2^{l_a-j_{k-1}-1}a\}.\end{align*}
Obviously, $C_{2^{l_a-j_l}a}^{(t)}\cap C_{2^{l_a-j_k}a}^{(t)}=\emptyset$ and $C_{2^{l_a-j_i}a}^{(t)}\subseteq C_a$ for any $i=1,2,\ldots,N+1$. In addition, we have   $\epsilon_{2^{l_a-j_k}a}^{(t)}=j_k-j_{k-1}$ and $\sum\limits_{k=1}^{N+1}\epsilon_{2^{l_a-j_k}a}^{(t)}=l_a$. The proof can be completed.
\end{proof}

The following lemmas will be useful in the sequel.
\begin{lemma}\label{coset}
  For any coset leader $j\in\Gamma\setminus\{0\}$, $j$ is odd and $1\leq j\leq 2^{m-1}-1$.
\end{lemma}
\begin{proof}
The proof can be obtained directly from \cite[Lemma~6]{si}.
\end{proof}
\begin{lemma}\label{1d}
	Suppose $\gcd(a,2^m-1)=1$ and $b\in\{aj:1\leq j\leq 2^{\frac{m}{2}}, 2\nmid j\}$ if $2\mid m$ and $b\in\{aj:1\leq j\leq 2^\frac{m+1}{2}-1, 2\nmid j\}$ if $2\nmid m$, then $l_b=m$.
\end{lemma}
\begin{proof}Obviously, $l_{2^{\frac{m}{2}}+1}=\frac{m}{2}$ and $l_{aj}=l_j$ for any $j\in\Z_{2^m-1}$. For any integer $j\in\{j:1\leq j\leq 2^{\frac{m}{2}}, 2\nmid j\}$ if $2\mid m$, or $\{j:1\leq j\leq 2^\frac{m+1}{2}-1, 2\nmid j\}$ if $2\nmid m$, suppose $l_j$ is the least integer such that $j2^{l_j}\equiv j\pmod{2^m-1}$, then
\begin{equation*}\frac{2^m-1}{2^{l_j}-1}\mid j\Longrightarrow j>2^{m-l_j}.\end{equation*}

Suppose $l_j\neq m$. Note that $l_j\leq \frac{m}{3}$ if $2\nmid m$, and $l_j\leq \frac{m}{2}$ if $2\mid m$, then\[j>2^{\frac{2m}{3}}\ \text{if}\ 2\nmid m,\ \text{and}\ j>2^{\frac{m}{2}}\ \text{if}\ 2\mid m, \] which are impossible since $j\leq 2^{\frac{m+1}{2}}-1$ if $2\nmid m$ and $j\leq 2^{\frac{m}{2}}$ if $2\mid m$.
\end{proof}

\section{Binary cyclic codes from two classes of  sequences}

\subsection{A generic construction of {cyclic codes from sequences} }
Given any polynomials $f(x)\in\F_{2^m}[x]$, Ding et al.\ \cite{Din2,Din3} constructed the following sequence
\begin{align*}&\bs=(s_t)_{t=0}^{n-1}=\bigl(\Tr_{2^m/2}(f(1+\alpha^t))\bigr)_{t=0}^{n-1}\end{align*}

By selecting some suitable {polynomials}   $f(x)$, they obtained some cyclic codes $\C_{\bs}$ with good parameters. In this paper, we will continue their work, and construct four classes of cyclic codes $\C_{\bs}$ by selecting $f(x)=x^{2^m-2}$ and $x+x^{2^m-2}+x^{2^h-1}$, where $0<h\leq \lceil\frac{m}{2}\rceil$.

For some $i\in\Gamma$, let \begin{equation}\label{vd}\rho_i=|\{j:j\in C_i, 2\mid j\}|\ \text{and}\ v_i\equiv\frac{m\rho_i}{l_i}\pmod 2.\end{equation} Clearly, $v_0=1$ for odd $m$, and $v_0=0$ for even $m$.

\vskip 0.2cm
\noindent
\textbf{A. Si-Ding sequence and its associated codes.} Take \[T_{(j,m)}=\{C_i: v_i\equiv j\pmod2,\ i\in\Gamma\}.\]

Clearly, $T_{(1,m)}\cup T_{(0,m)}=\Z_{2^m-1}$, $C_0=\{0\}\in T_{(1,m)}$ if $2\nmid m$, and $C_0\in T_{(0,m)}$ if $2\mid m$.
By Proposition~\ref{lem2}, we have the following result.
\begin{lemma}\label{2d}If $2\mid l_a$, then the $2$-cyclotomic cosets $C_a$ and $C_{a^c}$ are all in $T_{(1,m)}$ or $T_{(0,m)}$. If $2\nmid l_a$, only one of $C_a$ and $C_{a^c}$ is in $T_{(1,m)}$.
\end{lemma}

 \begin{definition}[\cite{si}]
   Take $f(x)=x^{2^m-2}$.
{Si-Ding sequence $\cS$} is 
\begin{equation}\label{siding}
\Tr_{2^m/2}((1+\alpha^t)^{2^m-2})=\bigl(\sum_{j\in\Gamma}v_j\sum_{i\in C_j}\alpha^{it}\bigr)_{t=0}^{n-1}=\bigl(\sum_{i\in T_{(1,m)}}\alpha^{it}\bigr)_{t=0}^{n-1}.
\end{equation}
 \end{definition}
They showed that the total number of nonzero coefficients of $\alpha^{it}$ in ~\eqref{siding} is $2^{m-1}$ and {the minimal polynomial $m_{\cS,1}(x)$} is\[\prod_{j\in T_{(1,m)}}(x-\alpha^j).\]

 The following lemma  shall provide a simpler proof on the size of $T_{(1,m)}$.

\begin{lemma}\label{lem5}
  With the same notations as above. Then $|T_{(1,m)}|=2^{m-1}$ and $|T_{(0,m)}|=2^{m-1}-1$.
  \end{lemma}
\begin{proof}
  Suppose {$a\in\Z_{2^m-1}\setminus\{0\}$} is an odd integer, and $l_a$ is the least integer such that $a2^{l_a}\equiv a\pmod{2^m-1}$. Then by ~\eqref{aq}, we have\[a=s_a+i_12s_a+
\cdots+i_{l_a-1}2^{l_a-1}s_a,\ i_j\in\{0,1\}.\] By Proposition~\ref{lem2}, we know
\[\wt_2(a)=(N+1)\frac{m}{l_a}\ \text{and}\ v_a=(l_a-N-1)\frac{m}{l_a}.\]

Set $N_{(i,m)}:=\{a\in\Z_{2^m-1}\setminus\{0\}:\wt_2(a)\equiv i\pmod2\}$. Then if $\frac{m}{l_a}$ is even, we know that $m$ is even and
\begin{equation}\label{e1}a\in N_{(0,m)}\ \text{and}\ a\in T_{(0,m)}.\end{equation}

If $\frac{m}{l_a}$ is odd, then we divide into two cases to consider.
\begin{itemize}
  \item $m$ is even, then $l_a$ is even and $\wt_2(a)\equiv v_a\pmod 2$. This means
  \begin{equation}\label{e2}a\in N_{(0,m)}\ \text{and}\ T_{(0,m)}\ \text{or}\ a\in N_{(1,m)}\ \text{and}\ T_{(1,m)}.\end{equation}
  \item $m$ is odd, then $l_a$ is odd, $\wt_2(a)\equiv N+1\pmod 2$ and $v_a\equiv N\pmod 2$. Therefore, we have
  \begin{equation}\label{e3}\wt_2(a^c)\equiv v_a\pmod2\ \text{and}\ \wt_2(a)\equiv v_{a^c}\pmod2.\end{equation}
\end{itemize}
 If $2\mid m$, then by Eqs.~\eqref{e1} and ~\eqref{e2}, we have
  \[|N_{(0,m)}|+1=|T_{(0,m)}|,\ \text{and}\ |N_{(1,m)}|=|T_{(1,m)}|.\]
 If $2\nmid m$, then by Eq.~\eqref{e3}, we obtain
    \[|N_{(0,m)}|+1=|T_{(1,m)}|,\ \text{and}\ |N_{(1,m)}|=|T_{(0,m)}|.\]

  From the definitions of $N_{(0,m)}$ and $N_{(1,m)}$, if $2\nmid m$, then
  \[|N_{(1,m)}|=\binom{m}{1}+\binom{m}{3}+\cdots+\binom{m}{m-2},\ |N_{(0,m)}|=\binom{m}{2}+\binom{m}{4}+\cdots+\binom{m}{m-1}.\]

  Note that $(1-1)^m=\sum\limits_{i=0}^{m}\binom{m}{i}(-1)^i=0$, $(1+1)^m=\sum\limits_{i=0}^{m}\binom{m}{i}=2^m$ and $a\in\Z_{2^m-1}\setminus\{0\}$, then
  $|N_{(1,m)}|=|N_{(0,m)}|=2^{m-1}-1$.

  If $2\mid m$, then
    \[|N_{(1,m)}|=\binom{m}{1}+\binom{m}{3}+\cdots+\binom{m}{m-1},\ |N_{(0,m)}|=\binom{m}{2}+\binom{m}{4}+\cdots+\binom{m}{m-2}.\]Then
    $|N_{(0,m)}|=2^{m-1}-2$, $|N_{(1,m)}|=2^{m-1}$. Therefore,
  \[|T_{(0,m)}|=2^{m-1}-1,\ \text{and}\ |T_{(1,m)}|=2^{m-1}.\]
  The proof can be completed.
\end{proof}

Next we construct the first class of code by using the sequence $\cS$ of ~\eqref{siding}. Take \begin{equation}\label{nws}M_{\cS,i}(x):=(x-1)^{\delta_{i,m}}m_{\cS,i}(x)=(x-1)^{\delta_{i,m}}\prod_{j\in T_{(i,m)}}(x-\alpha^j),\end{equation}where $\delta_{i,m}=1$ if $i=1$ and $2\mid m$, otherwise $\delta_{i,m}=0$, $i=0,1$. Clearly, $M_{\cS,1}(x)=(x-1)m_{\cS,1}(x)$ if $i=1$ and $2\mid m$, and the polynomial $M_{\cS,i}$ can be obtained directly from the sequence $\cS$ defined by ~\eqref{siding} if $2\nmid m$ or $2\mid m$ and $i=0$.
\begin{theorem}\label{dinth}Let
$M_{\cS,i}(x)$ be defined by ~\eqref{nws} and $\C_{\cS,i}=\langle M_{\cS,i}(x)\rangle$, where $i=0,1$. Then $\dim(\C_{\cS,i})=2^m-1-\deg(M_{\cS,i}(x))$ and $d(\C_{\cS,i})$ is even.
\end{theorem}
\begin{proof}
 The dimension of $\C_{\cS,i}$ is clear, and $d(\C_{\cS,i})$ is even since $x-1$ is a divisor of $M_{\cS,i}(x)$.
\end{proof}

 We next show that the code $\C_{\cS,i}$ is subcode of the code $\C_{i,m}$, which is constructed by Tang et al.\ \cite{Tang1} and has
 the generator polynomial
\begin{equation}\label{sdc}g_{i,m}(x)=\mathop{\prod}_{\substack{1\leq j\leq n-1 \\ \wt_2(j)\equiv i\mod2}}(x-\alpha^j)=\mathop{\prod}_{\substack{1\leq j\leq n-1 \\ j\in N_{(i,m)}}}(x-\alpha^j),\ i\in\{0,1\}.\end{equation}

By Lemma~\ref{lem5}, Eqs.~\eqref{nws} and ~\eqref{sdc}, we can see that
\[T_{(1,m)}=N_{(0,m)}\cup\{0\}\ \text{and}\  T_{(0,m)}=N_{(1,m)}\ \text{if}\ 2\nmid m \] and
\[T_{(0,m)}=N_{(0,m)}\cup\{0\}\ \text{and}\  T_{(1,m)}=N_{(1,m)}\cup\{0\}\ \text{if}\ 2\mid m. \]
Therefore, the following result can be obtained directly.
\begin{theorem}\label{sdct}Let
	$g_{i,m}(x)$ be defined by ~\eqref{sdc} and $\C_{i,m}=\langle g_{i,m}(x)\rangle$, where $i=0,1$. Then $\dim(\C_{i,m})=2^m-1-\deg(g_{i,m}(x))=2^m-1-|N_{(i,m)}|$. Furthermore, \[\C_{\cS,0}=\C_{1,m}\ \text{and}\ \C_{\cS,1}=\C_{1,m}^\bot\subset\C_{0,m}\ \text{if}\ 2\nmid m,\ \text{and}\]\[ \C_{\cS,1}=\C_{0,m}^\bot\subset\C_{1,m}\ \text{and}\ \C_{\cS,0}=\C_{1,m}^\bot\subset\C_{0,m}\ \text{if}\ 2\mid m.\]
\end{theorem}
When $m$ is odd. In \cite[Theorem~4.2]{Din2}, Ding showed that the minimum distances of the codes $\C_{\cS,1}$ and $\C_{\cS,1}^\bot$ satisfy $d^2\geq n$.  {In \cite[Theorems~14 and~15]{Tang1}}, Tang and Ding proved that \[d(\C_{0,m}),d(\C_{1,m})\geq\begin{cases}
                               2^{\frac{m-1}{2}}+1, & \mbox{if } m\equiv3\pmod4\geq3 \\
                               2^{\frac{m-1}{2}}+3, & \mbox{if } m\equiv1\pmod4\geq5.
                             \end{cases}\ \text{and}\]
\[d(\C^\bot_{0,m}),d(\C^\bot_{1,m})\geq\begin{cases}
                               2^{\frac{m-1}{2}}+2, & \mbox{if } m\equiv3\pmod4\geq3 \\
                               2^{\frac{m-1}{2}}+4, & \mbox{if } m\equiv1\pmod4\geq5.
                             \end{cases}\qquad  \]
                             
By Theorem~\ref{sdct}, the minimum distances of $\C_{\cS,1}$ and $\C_{\cS,0}$ also satisfy the above bounds.

When $m$ is even. Ding \cite{Din2} thought that the code may not have a good minimum distance by giving an $[15,7,3]$ code, and Tang and Ding \cite{Tang1} gave two examples $\C^\bot_{0,6}$ and $\C^\bot_{1,6}$ with parameters $[63,30,6]$ and $[63,32,10]$ respectively. We find that the minimum distances of $\C^\bot_{0,6}$ and $\C^\bot_{1,6}$ are bigger than their lower bounds presented by \cite[Theorem~13]{Tang1}, and this has been confirmed by checking more examples. Therefore, we will give a tight lower bounds of $\C_{\cS,1}$ and $\C_{\cS,0}$ when $m$ is even.

\vskip 0.2cm
\noindent
\textbf{B. Ding-Zhou sequence and {its associated codes}.} Take $f(x)=x+x^{2^m-2}+x^{2^h-1}$, $0< h\leq \lceil\frac{m}{2}\rceil$ . Let $\Gamma_1=\{a:1\leq a\leq 2^h-1, 2\nmid a\}$ and $\kappa_a^{(h)}\equiv \epsilon_a^{(h)}\pmod2$.  We define
\[u_j=\begin{cases}
	\kappa_1^{(h)}+v_1+1\pmod2, & \mbox{if } j=1 \\
	(\kappa_j^{(h)}+v_j)\pmod2, & \mbox{if } j\in\Gamma_1\setminus\{1\}\\
	v_j, & \mbox{if } j\in\Gamma\setminus\Gamma_1.
\end{cases}\]
Note that  $\kappa_1^{(h)}=h$, then
\begin{equation}\label{dz}u_j=\begin{cases}
	v_1+h+1\pmod2, & \mbox{if } j=1 \\
	(\kappa_j^{(h)}+v_j)\pmod2, & \mbox{if } j\in\Gamma_1\setminus\{1\}\\
	v_j, & \mbox{if } j\in\Gamma\setminus\Gamma_1.
\end{cases}\end{equation}
\begin{definition}
  Take $f(x)=x+x^{2^m-2}+x^{2^h-1}$. {Ding-Zhou sequence $\cD$} is 
\begin{equation}\label{siding1}
(\Tr_{2^m/2}(f(1+\alpha^t)))_{t=0}^{n-1}
=(\sum_{j\in \Gamma}u_j\sum_{i\in C_j}\alpha^{ti})_{t=0}^{n-1}.
\end{equation}
\end{definition}Take $D_{(j,m)}=\{C_i:u_i\equiv j \pmod2,i\in\Gamma\}$, where $j=0,1$.
By ~\eqref{siding1}, the minimal polynomial $m_{\cD,1}(x)$ of $\cD$ is
\begin{equation}\label{sigen1}
	 \prod\limits_{j\in D_{(1,m)}}(x-\alpha^j).
\end{equation}
\begin{lemma}
With the same notations as above. Then
$|T_{(j,m)}|=|D_{(j,m)}|$, $j=0,1$.
\end{lemma}
\begin{proof}
  For {$j=1$}, we have $v_1=(m-1)\pmod2$ and $u_1=(m+h)\pmod2$. For $j=3$, we have $v_3=m\pmod2$ and $u_3=(m+h-1)\pmod2$. {Independently of the parity of $m$, there are exactly one nonzero $v_j$ and $u_j$ for $j\in\{1,3\}$}. Note that \[\Gamma_1=\bigcup_{a=2}^{h-1}\{k:2^a<k<2^{a+1},2\nmid k\}\cup\{1,3\}.\]
 We have $|\{k:2^a<k<2^{a+1},2\nmid k\}|=2^{a-1}$ and \[|\{k:2^a<k<2^{a+1},2\nmid k,\wt_2(k)\equiv0\pmod2\}|=2^{a-2}.\] Therefore, by Lemma~\ref{1d}, for any odd integer $k$ with $2^a<k<2^{a+1}$, we have
 \[|T_{(0,m)}|=|T_{(1,m)}|=|D_{(0,m)}|=|D_{(1,m)}|=2^{a-2}.\]
It then { follows} that
\[|\{j\in\Gamma_1:v_j\equiv i\pmod2\}|=|\{j\in\Gamma_1:u_j\equiv i\pmod2\}|.\]Therefore, by Eq.~\eqref{dz}, the desired result can be obtained.
\end{proof}

Next we construct the second class of code by using the sequence $\cD$ of ~\eqref{siding1}. Take
\begin{equation}\label{degm}
	M_{\cD,i}=(x-1)^{\delta_{m,i}}\prod\limits_{j\in D_{(i,m)}}(x-\alpha^j),\end{equation}
where $\delta_{i,m}=1$ if $i=1$ and $2\mid m$, otherwise $\delta_{i,m}=0$.

\begin{theorem}
	Let $M_{\cD,i}(x)$ be defined by ~\eqref{degm} and $\C_{\cD,i}=\langle M_{\cD,i}(x)\rangle$. Then $\C_{\cD,i}$ has length $2^m-1$ and dimension $2^{m}-1-\deg(M_{\cD,i}(x))$.
\end{theorem}
\begin{proof}
Clear.
\end{proof}

%For $2\nmid m$ and $h\neq0$, if $j\in\Gamma_1$ and $\wt_2(j)=h$, then $\wt_2(2^m-1-h)=m-h$, i.e., there are exactly one nonzero $u_j$ and $u_{j^c}$ for any $j\in\Gamma_1$. Set $\cD_{(j,m)}=\{C_i:u_i\equiv j\pmod2,i\in\Gamma\}$ and $\cD^*_{(1,m)}:=\cD_{(1,m)}-\{0\}$. Then
% we have \begin{itemize}
 %	\item $\cD^*_{(1,m)}\cap \cD_{(0,m)}=\emptyset$ and $\cD^*_{(1,m)}\cup \cD_{(0,m)}=\Z_n-\{0\}$,
 %	\item $\cD^*_{(1,m)}=-\cD_{(0,m)}$.
% \end{itemize}
%\begin{corollary}\label{cor3} Take $\C'_{\cD,1}=\langle M_{\cD,1}(x)/(x-1)\rangle$ and $\C'_{\cD,0}=\langle (x-1)M_{\cD,0}(x)\rangle$. Then
 % \[\dim(\C_{\cD,0})=\dim(\C'_{\cD,1})=2^{m-1},\ \text{and}\ \dim(\C_{\cD,1})=\dim(\C'_{\cD,0})=2^{m-1}-1.\]
%\end{corollary}
%
% The pair of binary cyclic codes $\C_{\cD,0}$ and $\C'_{\cD,1}$ are called \emph{odd-like duadic codes}, $\C_{\cD,1}$ and $\C'_{\cD,0}$ are called \emph{even-like duadic codes}.

\section{Lower bounds on the minimum distance of $\C_{\cS,i}$}
In this section, we let $m$ be even and then study {the lower bounds on the minimum distance} of $\C_{\cS,i}$.
We firstly give some auxiliary results.

\subsection{Some auxiliary results}

\begin{lemma}\label{lemsd1}
	Set $m\equiv0\pmod4\geq4$ and $a=2^{\frac{m+2}{2}}-1$. Then
	\[ \{ak:1\leq k\leq 2^{\frac{m-2}{2}}\}\cup\{ak:2^m-2^{\frac{m-2}{2}}-1\leq k\leq 2^{m}-2\}\subseteq T_{(1,m)}.\]
\end{lemma}
\begin{proof}
	By Lemma~\ref{lem2}, the $2$-cyclotomic cosets containing $a$ and $n-a$ are in $T_{(1,m)}$. Hence it is sufficient to prove $\{a k:1\leq k\leq 2^{\frac{m-2}{2}},2\nmid k\}\subseteq T_{(1,m)}$.
	
	Note that \[\gcd(a,2^m-1)=2^{\gcd(\frac{m+2}{2},\frac{m-2}{2})}-1=2^{\gcd(\frac{m+2}{2},2)}-1=1.\] By Lemma~\ref{1d},
we then have $l_{ak_1}=m$ for any $k_1\in\{k:1\leq k\leq 2^{\frac{m-2}{2}},2\nmid k\}$. By Proposition~\ref{lem2}, we only show that \[\wt_2(ak_1)=\frac{m+2}{2}\ \text{for any}\ k_1\in\{k:1\leq k\leq 2^{\frac{m-2}{2}}-1,2\nmid k\}.\]

Suppose
	\[k_1=1+2^{i_{1,1}}+2^{i_{2,1}}+\cdots+2^{i_{l,1}}, \] where $1\leq l\leq \frac{m-4}{2}$ and $1\leq i_{1,1}<i_{2,1}<\cdots<i_{l,1}\leq \frac{m-4}{2}$.
	Then
	\begin{align*}
		ak_1 &= 2^{\frac{m+2}{2}}k_1-k_1=\sum_{j=1}^{l}2^{\frac{m+2}{2}+i_{j,1}}+2^{\frac{m+2}{2}}-1-
		\sum_{j=1}^{l}2^{i_{j,1}} \\
		&=\sum_{j=1}^{l}2^{\frac{m+2}{2}+i_{j,1}}+2^{\frac{m}{2}}+2^{\frac{m-2}{2}}+
		\sum_{j\in\{1,2,\ldots,\frac{m-4}{2}\}\setminus\{i_{j,1},\ldots,i_{l,1}\}}2^j+1.
	\end{align*}Clearly, $\wt_2(ak_1)=\frac{m+2}{2}$, $\rho_{ak_1}=\frac{m-2}{2}\equiv1\pmod2$. Hence the proof can be completed.
\end{proof}

\begin{lemma}\label{lemsd2}
	Set $m\equiv2\pmod4\geq6$ and $a=2^{\frac{m+4}{2}}-1$. Then
	\[ \{ak:1\leq k\leq 2^{\frac{m-4}{2}}\}\cup\{a k:2^m-2^{\frac{m-4}{2}}-1\leq k\leq 2^{m}-2\}\subseteq T_{(1,m)}.\]
\end{lemma}
\begin{proof}
	Note that \[\gcd(a,2^m-1)=2^{\gcd(\frac{m+4}{2},\frac{m-4}{2})}-1=2^{\gcd(\frac{m+4}{2},4)}-1=1.\]

By using the same approach of Lemma~\ref{lemsd1}, we also have $\wt_2(ak_1)=\frac{m+4}{2}$ and $\rho_{ak_1}=\frac{m-4}{2}\equiv1\pmod2$ for any $k_1\in\{k:1\leq k\leq 2^{\frac{m-4}{2}}-1,2\nmid k\}$.
\end{proof}
\begin{lemma}\label{lemsd3}
	Let $l\geq1$, $2\nmid e>1$ and $m=2^le\geq4$, let $a=2^{\frac{m+2^l}{2}}+1$. Then
\[ \{ka:1\leq k\leq 2^{\frac{m-2^l}{2}}\}\cup\{ka:2^m-2^{\frac{m-2^l}{2}}-1\leq k\leq 2^{m}-1\}\subseteq T_{(0,m)}.\]
\end{lemma}
\begin{proof}
 Note that $\gcd(\frac{m+2^l}{2},m)=\gcd(\frac{m+2^l}{2},2^l)=2^l$. Then $\gcd(2^{\frac{m+2^l}{2}}+1,2^m-1)=1$.
	
	For integer $k_1\in\{1\leq k\leq 2^{\frac{m-2^l}{2}}-1,2\nmid k\}$, set
	\[k_1=1+2^{i_{1,1}}+2^{i_{2,1}}+\cdots+2^{i_{l,1}},\]where $1\leq l\leq \frac{m-2^l-2}{2}$ and $1\leq i_{1,1}<i_{2,1}<\cdots<i_{l,1}\leq \frac{m-2^l-2}{2}$. Clearly, $\wt_2(ak_1)=2\wt_2(k_1)$, $\rho_{ak_1}\equiv0\pmod2$ for any $k_1\in\{1\leq k\leq 2^{\frac{m-2^l}{2}},2\nmid k\}$.
\end{proof}

\subsection{ {Lower bounds} on the minimum distance}
\begin{theorem}\label{1}Let $m=2^le>2$, where $2\nmid e$ and $l\geq1$, let $\C_{\cS,i}=\langle M_{\cS,i}(x)\rangle$ be defined by ~\eqref{nws}. Then $\C_{\cS,1}$ has parameters $[2^m-1,2^{m-1}-2,d(\C_{\cS,1})]$, where
	\[d(\C_{\cS,1})\geq \begin{cases}
		2^{\frac{m}{2}}+2, & \mbox{if } l\geq2, \\
		2^{\frac{m-2}{2}}+2, & \mbox{if } l=1 .
	\end{cases}\]
Furthermore, if $e\geq3$, then	$\C_{\cS,0}$ has parameters $[2^m-1,2^{m-1},\geq
	2^{\frac{m-2^l+2}{2}}+2]$. In particular, if $l=1$, then
	\[d(\C_{\cS,0})\geq 2^{\frac{m}{2}}+2.\]
\end{theorem}
\begin{proof}
	Note that in Lemmas~\ref{lemsd1}-\ref{lemsd3}, $\gcd(a,2^m-1)=1$, so there  exists an integer $a'$ such that $aa'\equiv1\pmod{2^m-1}$. Take $\gamma=\alpha^{a'}$. Then $\gamma$ is also an $n$-th primitive root of unity.
	
	For $m\equiv0\pmod4$, by Lemmas~\ref{lemsd1} and ~\ref{lemsd3}, the defining sets $T_{(1,m)}$ and $T_{(0,m)}$ with respect to $\gamma$ {contain the sets} \[\{k:1\leq k\leq 2^{\frac{m-2}{2}}\}\cup\{k:2^m-2^{\frac{m-2}{2}}-1\leq k\leq 2^{m}-2\}\cup\{0\}\] and
	\[ \{k:1\leq k\leq 2^{\frac{m-2^l}{2}}\}\cup\{k:2^m-2^{\frac{m-2^l}{2}}-1\leq k\leq 2^{m}-1\}.\]
	
	 Hence the desired lower bounds follow from the BCH bound. The case of $m\equiv2\pmod4$ can be obtained similarly, we omit it here.
\end{proof}
\begin{remark}
In \cite[Theorem~13]{Tang1}, Tang and Ding proved that
\[d(\C_{0,m}^\bot)\geq\begin{cases}
                        2^{\frac{m-2}{2}}+2, & \mbox{if } l\geq2 \\
                        2^{\frac{m-4}{2}}+2, & \mbox{if } l=1.
                      \end{cases}\ \text{and}\ d(\C_{1,m}^\bot)\geq \begin{cases}
                                               2^{\frac{m-4}{2}}+2, & \mbox{if } l\geq2 \\
                                               2^{\frac{m-2}{2}}+2, & \mbox{if } l=1.
                                             \end{cases}\]
By Theorem~\ref{sdct}, we can see that the lower bounds in {Theorem~\ref{1}} are better than them.
\end{remark}

By {Theorem~\ref{1}}, we can see that the lower bounds of $d(\C_{\cS,1})$ $(l\geq2)$ and $d(\C_{\cS,0})$ $(l=1)$ meet $d^2\geq n$. For the exact values there, based on numerical evidence in Table~1, we form the following conjecture.
\begin{conj}
Take $\C_{\cS,m,i}:=\C_{\cS,i}$. If $l\geq2$, then $d(\C_{\cS,m,1})=d(\C_{\cS,m+2,1})=2^{\frac{m}{2}}+2$. If $l=1$ and $m+2$ is not $2$-power, then $d(\C_{\cS,m,0})=d(\C_{\cS,m+2,0})=2^{\frac{m}{2}}+2$, and $d(\C_{\cS,m,1})=2^{\frac{m-2}{2}}+2$.
\end{conj}

\begin{table}[h]
\caption{\textbf{The minimum distances of $\C_{\cS,1}$ and $\C_{\cS,0}$}}%标题
\centering%把表居中
\label{table:t1}
\begin{tabular}{|c|cc|cc|}%五个c代表该表一共五列，内容全部居中
\hline%第一道横线
\multirow{2}{*}{$m$} & \multirow{2}{*}{$\dim(\C_{\cS,1})$} & \multirow{2}{*}{$d(\C_{\cS,1})$} & \multirow{2}{*}{$\dim(\C_{\cS,0})$}&\multirow{2}{*}{$d(\C_{\cS,0})$} \\
&  & & & \\
\hline%第二道横线
\multirow{2}{*}{4}&\multirow{2}{*}{6}&\multirow{2}{*}{6}& \multirow{2}{*}{8} &\multirow{2}{*}{4}\\
& & & &\\
\hline%第三道横线
\multirow{2}{*}{6}&\multirow{2}{*}{30}&\multirow{2}{*}{6}&\multirow{2}{*}{32}&\multirow{2}{*}{10} \\
& & & &\\
\hline%第四道横线
\multirow{2}{*}{8}&\multirow{2}{*}{126}&\multirow{2}{*}{18}&\multirow{2}{*}{128}&\multirow{2}{*}{22} \\
& & & &\\
\hline%第五道横线
\multirow{2}{*}{10}&\multirow{2}{*}{510}&\multirow{2}{*}{18} &\multirow{2}{*}{512}&\multirow{2}{*}{34} \\
& & & & \\
\hline%第六道横线
\multirow{2}{*}{12}&\multirow{2}{*}{2046}&\multirow{2}{*}{66}&\multirow{2}{*}{2048}&\multirow{2}{*}{34} \\
&  & &     &\\
\hline%第七道横线
\multirow{2}{*}{14}&\multirow{2}{*}{8190}&\multirow{2}{*}{66}&\multirow{2}{*}{8192}&\multirow{2}{*}{130} \\
& & & &\\
\hline%第八道横线
\multirow{2}{*}{16}&\multirow{2}{*}{32766}&\multirow{2}{*}{258} &\multirow{2}{*}{32768}&\multirow{2}{*}{146} \\
& & & &\\
\hline%第九道横线
\multirow{2}{*}{18}&\multirow{2}{*}{131070}&\multirow{2}{*}{258}&\multirow{2}{*}{131072}&\multirow{2}{*}{514}  \\
& & & &\\
\hline%第十道横线
\multirow{2}{*}{20}&\multirow{2}{*}{524286}&\multirow{2}{*}{1026}&\multirow{2}{*}{524288}&\multirow{2}{*}{514}  \\
& & & &\\
\hline%第十一道横线
\multirow{2}{*}{22}&\multirow{2}{*}{2097150}&\multirow{2}{*}{1026}&\multirow{2}{*}{2097152}&\multirow{2}{*}{2050} \\
& & & &\\
\hline%第十二道横线
\multirow{2}{*}{24}&\multirow{2}{*}{8388606}&\multirow{2}{*}{4098}&\multirow{2}{*}{8388608}&\multirow{2}{*}{2050} \\
& & & &\\
\hline%第十三道横线
\multirow{2}{*}{26}&\multirow{2}{*}{33554430}&\multirow{2}{*}{4098}&\multirow{2}{*}{33554432}&\multirow{2}{*}{8194} \\
& & & &\\
\hline%第十四道横线
\end{tabular}
\end{table}
From Table~\ref{table:t1}, we can see that the code $\C_{\cS,0}$ may not have a good minimum distance when $m$ is $2$-power. For example, if $m=8$,  $d(\C_{\cS,0})=22$ {which satisfies $d(\C_{\cS,0})^2\geq n$}. If $m=16$, $d(\C_{\cS,0})=146$ which can not satisfy $d(\C_{\cS,0})^2\geq n$. So we are not interested in the case that $m$ is $2$-power.

%\subsection{ Comparisons $\C_{\cS,1}$ and $\C_{\cS,0}$ with the known codes}

%When $m\geq6$ is even, the length $2^m-1$ can not be the product of two distinct primes. Then the code $\C_{\cS,0}$ can not be any of the six families of binary codes presented in \cite{Din5}. But when $m=4$, the defining sets of the codes $\C_{\cS,0}$ and $\C_{\cS,1}$ are different from the codes presented in \cite{Din5}.

%In fact, by \cite[Theorem~13]{tangw}, we can see that the codes $\C_{\cS,0}$ and $\C_{\cS,1}$ are the even-weight subcodes of the codes $\C_{0,m}$ and $\C_{1,m}$, while the lower bound presented in Theorem~\ref{1} is bigger than \cite[Theorem~13]{tangw}.% since our defining sets of the codes $\C_{\cS,0}$ and $\C_{\cS,1}$ contain the set $\{0\}$, but the codes $\C_{0,m}$ and $\C_{1,m}$ do not have it. In addition, the codes $\C_{\cS,0}$ and $\C_{\cS,1}$ could be much better than the codes $\C_{0,m}$ and $\C_{1,m}$. For example, when $m=4$, the code $\C_{1,4}$ has parameters  $[15,7,3]$ , while the code $\C_{\cS,1}$ has parameters $[15,6,6]$. When $m=6$,  the code $\C_{0,6}$ has parameters  $[63,33,7]$ , while the code $\C_{\cS,0}$ has parameters $[63,32,10]$.

\section{Lower bounds on the minimum distance of $\C_{\cD,i}$}

In this section, we shall divide into two cases to determine the lower bounds on the minimum distances of the codes $\C_{\cD,0}$ and $\C_{\cD,1}$.

\subsection{ Some auxiliary results}

\begin{lemma}\label{lem3}
	Let $m\equiv1\pmod4\geq5$, $0< h\leq \frac{m-3}{2}$ and $a=2^{\frac{m-1}{2}}-1$. Then
	\begin{itemize} \item $\{ak:0\leq k\leq 2^{\frac{m-1}{2}}+2\}\subseteq D_{(1,m)}$,\\
		\item $\{ak:2^m-2^{\frac{m-1}{2}}-3\leq k\leq 2^m-2\}\subseteq D_{(0,m)}$ if $2\nmid h$ and $\{ak:2^m-2^{\frac{m-1}{2}}-1\leq k\leq 2^m-2\}\subseteq D_{(0,m)}$ if $2\mid h$.\end{itemize}
\end{lemma}
\begin{proof}Note that $\gcd(a,2^m-1)=2^{\gcd(\frac{m+1}{2},\frac{m-1}{2})}-1=1$, by Lemma~\ref{1d}, we then know that $|C_{ak_1}|=m$ for any $k_1\in\{k:1\leq k\leq 2^{\frac{m-1}{2}}-1,2\nmid k\}$.
	
	For $k=2^{\frac{m-1}{2}}+1$, we have
	\begin{equation}\label{eq:e}ak\equiv 2^{m-1}-1\pmod{2^m-1}.\end{equation}
	Then $\wt_2(2^{m-1}-1)=m-1$, $\rho_{2^{m-1}-1}=1$. For $k=0$, it is clear that $\{0\}\in D_{(1,m)}$.
	
	We next prove that $\wt_2(ak_1)=\frac{m-1}{2}$ for any $k_1\in\{k:1\leq k\leq 2^{\frac{m-1}{2}}-1,2\nmid k\}$. Set \[k_1=1+2^{i_{1,1}}+2^{i_{2,1}}+\cdots+2^{i_{l,1}},\] where $1\leq l\leq \frac{m-3}{2}$ and $1\leq i_{1,1}<i_{2,1}<\cdots<i_{l,1}\leq \frac{m-3}{2}$.
	Then
	\begin{align}
		ak_1 &= 2^{\frac{m-1}{2}}k_1-k_1=\sum_{j=1}^{l}2^{\frac{m-1}{2}+i_{j,1}}
		+2^{\frac{m-1}{2}}
		-1-\sum_{j=1}^{l}2^{i_{j,1}}\notag \\
		&=\sum_{j=1}^{l}2^{\frac{m-1}{2}+i_{j,1}}
		+\sum_{j\in\{1,2,\ldots,\frac{m-3}{2}\}\setminus\{i_{1,1},\ldots,i_{l,1}\}}2^{j}+1.\label{qw}
	\end{align}
	Clearly, $\wt_2(ak_1)=\frac{m-1}{2}$. Hence $\rho_{ak_1}=\frac{m+1}{2}$.
	
	Since $\max{\Gamma_1}\leq 2^{\frac{m-3}{2}}-1$, then $\wt_2(k_2)\leq \frac{m-3}{2}$ for any $k_2\in\Gamma_1$, and \[2^sak_1\notin\Gamma_1\ \text{for any}\ k_1\in \{k:1\leq k\leq 2^{\frac{m-1}{2}},2\nmid k\}.\]
	
	From Eq.~\eqref{eq:e}, we know $ak_1=2^{m-1}-1\notin\Gamma_1$ for $k=2^{\frac{m-1}{2}}+1$.
	Therefore, $ak_1\in\Gamma\setminus\Gamma_1$ for any $k_1\in\{k:1\leq k\leq 2^{\frac{m-1}{2}}+2,2\nmid k\}$, i.e., \[\{ak:0\leq k\leq 2^{\frac{m-1}{2}}+2\}\subseteq D_{(1,m)}.\]
	
	In addition, by ~\eqref{qw}, we have
	 $\wt_2(2^m-1-ak_1)=\frac{m+1}{2}$ and $\rho_{2^m-1-ak_1}=\frac{m-1}{2}$ and \[2^s(2^m-1-ak_1)\notin\Gamma_1\ \text{for any}\ k_1\in \{k:1\leq k\leq 2^{\frac{m-1}{2}},2\nmid k\}.\]
	Hence, $\{ak:2^m-2^{\frac{m-1}{2}}-1\leq k\leq 2^m-2, 2\nmid k\}\subseteq D_{(0,m)}$.
	
	For $k_1=2^{\frac{m-1}{2}}+1$, by Eq.~\eqref{eq:e}, we know $2^m-1-ak_1=2^{m-1}\in C_1$. Then  by ~\eqref{dz}, we have
	\[u_1=h+1\pmod2\equiv 0\ \text{if}\ 2\nmid h,\ \equiv1\ \text{if}\ 2\mid h.\]
	
	Therefore, the proof can be completed.
\end{proof}

\begin{lemma}\label{lem30}
	Let $m\equiv3\pmod4\geq7$, $0< h\leq \frac{m-3}{2}$ and $a=2^{\frac{m+1}{2}}-1$. Then
	\begin{itemize}\item $\{ak:0\leq k\leq 2^{\frac{m-1}{2}}+2\}\subseteq D_{(1,m)}$ if $2\mid h$ and $\{ak:0\leq k\leq 2^{\frac{m-1}{2}}\}\subseteq D_{(1,m)}$ if $2\nmid h$,
		\item  $\{ak:2^m-2^{\frac{m-1}{2}}-1\leq k\leq 2^m-2\}\subseteq D_{(0,m)}$.\end{itemize}

\end{lemma}
\begin{proof}
	Note that $\gcd(a,2^m-1)=2^{\gcd(\frac{m+1}{2},\frac{m-1}{2})}-1=1$, by Lemma~\ref{1d}, we have $|C_{ak_1}|=m$ for any $k_1\in\{k:1\leq k\leq 2^{\frac{m-1}{2}}-1,2\nmid k\}$.
	
	 Set
	\[k_1=1+2^{i_{1,1}}+2^{i_{2,1}}+\cdots+2^{i_{l,1}},\] where $1\leq l\leq \frac{m-3}{2}$ and $1\leq i_{1,1}<i_{2,1}<\cdots<i_{l,1}\leq \frac{m-3}{2}$. Then
	\begin{align}
		ak_1 &= 2^{\frac{m+1}{2}}k_1-k_1=\sum_{j=1}^{l}2^{\frac{m+1}{2}+i_{j,1}}
		+2^{\frac{m+1}{2}}
		-1-\sum_{j=1}^{l}2^{i_{j,1}}\notag \\
		&=\sum_{j=1}^{l}2^{\frac{m+1}{2}+i_{j,1}}+2^{\frac{m-1}{2}}		+\sum_{j\in\{1,2,\ldots,\frac{m-3}{2}\}\setminus\{i_{1,1},\ldots,i_{l,1}\}}2^{j}+1.\label{qw1}
	\end{align}

	By ~\eqref{qw1}, we have $\wt_2(ak_1)=\frac{m+1}{2}$ and  $\wt_2(-ak_1)=\frac{m-1}{2}$. Hence $\rho_{ak_1}=\frac{m-1}{2}\equiv 1\pmod2$ and $\rho_{-ak_1}=\frac{m+1}{2}\equiv0\pmod2$.
	
	Note that $\max\Gamma_1\leq 2^{\frac{m-3}{2}}-1$, then $\wt_2(k_2)\leq \frac{m-3}{2}$ for any $k_2\in\Gamma_1$ and \[2^sak_1,-2^sak_1\notin\Gamma_1\ \text{for any}\ k_1\in \{k:1\leq k\leq 2^{\frac{m-1}{2}},2\nmid k\}.\]
	
	Therefore, $\{ak:1\leq k\leq 2^{\frac{m-1}{2}}\}\subseteq D_{(1,m)}$ and $\{ak:2^m-2^{\frac{m-1}{2}}-1 \leq k\leq 2^m-2\}\subseteq D_{(0,m)}$.

	For $k_1=2^{\frac{m-1}{2}}+1$, we have
	\[ak_1=(2^{\frac{m+1}{2}}-1)(2^{\frac{m-1}{2}}+1)\equiv 2^{\frac{m-1}{2}}\pmod{2^m-1}.\]
	Then $\wt_2(ak_1)=1$ and $\wt_2(-ak_1)=m-1$. By Eq.~\eqref{dz}, we have
	\[u_1=h+1\pmod2\equiv 1\ \text{if}\ 2\mid h,\ \text{and}\equiv0\ \text{if}\  2\nmid h.\]
	This means $ C_1\in D_{(1,m)}$ if $2\mid h$, and $ C_1\notin D_{(1,m)}$ if $2\nmid h$.
	Therefore,
	 the proof can be completed.
\end{proof}

\begin{lemma}\label{lemdz10}
	Let $m\equiv0\pmod4\geq8$, $0< h\leq \frac{m-4}{2}$ and $a=2^{\frac{m+2}{2}}-1$. Then
	\begin{itemize} \item $\{ak:1\leq k\leq 2^{\frac{m-2}{2}}\}\cup\{ak:2^m-2^{\frac{m-2}{2}}-1\leq k\leq 2^{m}-2\}\subseteq D_{(1,m)}$ if $2\nmid h$,
		\item $\{ak:1\leq k\leq 2^{\frac{m-2}{2}}+2\}\cup\{ak:2^m-2^{\frac{m-2}{2}}-1\leq k\leq 2^{m}-2\}\subseteq D_{(1,m)}$ if $2\mid h$.\end{itemize}
\end{lemma}
\begin{proof}
	By Lemma~\ref{lemsd1}, we know $\wt_2(ak_1)=\frac{m+2}{2}$ and $\wt_2(-ak_1)=\frac{m-2}{2}$ for any $\{k:1\leq k\leq 2^{\frac{m-2}{2}},2\nmid k\}$.
	Note that $\wt_2(k_2)\leq \frac{m-4}{2}$ for any $k_2\in\Gamma_1$. Then
	\[2^sak_1,-2^sak_1\notin\Gamma_1\  \text{for any}\ k_1\in \{k:1\leq k\leq 2^{\frac{m-2}{2}},2\nmid k\}.\]
	
	Therefore, we have $\{ak:1\leq k\leq 2^{\frac{m-2}{2}}\}\cup\{ak:2^m-2^{\frac{m-2}{2}}-1\leq k\leq 2^{m}-2\}\subseteq D_{(1,m)}$.
	
	For $k_1=2^{\frac{m-2}{2}}+1$, we have
	\begin{align*}ak_1\pmod{2^m-1}&\equiv 2^{\frac{m}{2}}+2^{\frac{m-2}{2}}\Rightarrow ak_1\in C_3,\\
2^m-1-ak_1\pmod{2^m-1}&\equiv 2^{m-2}-1\notin\Gamma_1.
	\end{align*}
	By Eq.~\eqref{dz}, we have \[u_3=m+h-1\pmod2\equiv1\ \text{if}\ 2\mid h,\quad \equiv0\ \text{if}\ 2\nmid h.\] Then $C_3\in D_{(1,m)}$ if $2\mid h$, $C_3\notin D_{(1,m)}$ if $2\nmid h$. Therefore, the desired results can be obtained.
\end{proof}

\begin{lemma}\label{lemdz11}
	Let $m\equiv2\pmod4\geq10$, $0< h\leq \frac{m-6}{2}$ and $a=2^{\frac{m+4}{2}}-1$. Then
	\begin{itemize}
		\item $\{ak:1\leq k\leq 2^{\frac{m-4}{2}}+2\}\cup\{ak:2^m-2^{\frac{m-4}{2}}-1\leq k\leq 2^{m}-2\}\subseteq D_{(1,m)}$ if $h\geq 4$ and $2\mid h$,
		\item $\{ak:1\leq k\leq 2^{\frac{m-4}{2}}\}\cup\{a k:2^m-2^{\frac{m-4}{2}}-1\leq k\leq 2^{m}-2\}\subseteq D_{(1,m)}$ if $0< h\leq3$ or $h\geq4$ and $2\nmid h$.
	\end{itemize}

\end{lemma}
\begin{proof}
	By Lemma~\ref{lemsd2}, it is sufficient to prove that
	\[2^sak_1,-2^sak_1\notin\Gamma_1\  \text{for any}\ k_1\in \{k:1\leq k\leq 2^{\frac{m-4}{2}},2\nmid k\}.\]
	The result is clear since $\wt_2(k_2)\leq\frac{m-6}{2}$ for any $k_2\in\Gamma_1$. Therefore, $\{ak:1\leq k\leq 2^{\frac{m-4}{2}}\}\cup\{a k:2^m-2^{\frac{m-4}{2}}-1\leq k\leq 2^{m}-2\}\subseteq D_{(1,m)}$.
	
	For $k_1=2^{\frac{m-4}{2}}+1$, we have
	\begin{align*}ak_1\pmod{2^m-1}&\equiv 2^{\frac{m+4}{2}}-2^{\frac{m-4}{2}}\in C_{15}
.\end{align*}
	Note that $v_{15}=m-4\equiv0\pmod2$. Then $C_{15}\notin D_{(1,m)}$ if $0< h\leq3$.
	
	If $h\geq 4$, then $15\in\Gamma_1$ and $u_{15}\equiv h-3\pmod{2}$. Hence \[C_{15}\in D_{(1,m)}\ \text{if}\ 2\mid h,\ \text{and}\ C_{15}\notin D_{(1,m)}\ \text{if}\  2\nmid h.\]Note that
\[-ak_1\pmod{2^m-1}\equiv 2^{m-4}-1\notin\Gamma_1,\ \text{and}\ v_{2^{m-2}-1}=2\equiv0\pmod2.\]  Therefore, the proof can be completed.
\end{proof}
\begin{lemma}\label{lemdz12}	Let $l\geq1$, $2\nmid e>1$ and $m=2^le\geq6$, let $0< h\leq 2^{l}$ and $a=2^{\frac{m+2^l}{2}}+1$. Then
	\begin{align*}\{ak:1\leq k\leq 2^{\frac{m-2^l}{2}}\}\cup\{ak:2^m-2^{\frac{m-2^l}{2}}-1\leq k\leq 2^{m}-1\}\subseteq D_{(0,m)}\ \text{if}\  h\neq 2^l,\\
	\{ak:1\leq k\leq 2^{\frac{m-2^l}{2}}\}\cup\{ak:2^m-2^{\frac{m-2^l}{2}}+1\leq k\leq 2^{m}-1\}\subseteq D_{(0,m)}\ \text{if}\  h= 2^l\end{align*}
\end{lemma}
\begin{proof}
	By Lemma~\ref{lemsd3},
	we need to prove that $2^sak_1,-2^sak_1\notin\Gamma_1$ for any $k_1\in\{k:1\leq k\leq 2^{\frac{m-2^l}{2}},2\nmid k\}$. Set
	\[k_1=1+2^{i_{1,1}}+2^{i_{2,1}}+\cdots+2^{i_{l,1}},\]where $1\leq l\leq \frac{m-2^l-2}{2}$ and $1\leq i_{1,1}<i_{2,1}<\cdots<i_{l,1}\leq \frac{m-2^l-2}{2}$.
	Then
	\begin{align*}
		ak_1 &= 2^{\frac{m+2^l}{2}}k_1+k_1=\sum_{j=1}^{l}2^{\frac{m+2^l}{2}+i_{j,1}}+2^{\frac{m+2^l}{2}}+
		\sum_{j=1}^{l}2^{i_{j,1}}+1. \end{align*}Note that
	\begin{align}&\min\{2^{\frac{m-2^l}{2}-i_{1,1}}ak_1,\ldots,
		2^{\frac{m-2^l}{2}}ak_1,2^{m-i_{1,1}}ak_1,\ldots,2^{m-i_{l,1}}ak_1\}
		> 2^{\frac{m-2^l}{2}}+1,\label{a0}\end{align}
	this means $2^sak_1\notin \Gamma_1$ for any $s>0$.  Set $U=\{i_1,i_2,\ldots,i_l\}$. Then
	\[-k_1\pmod{2^m-1}=\sum_{j=\frac{m-2^l}{2}}^{m-1}2^j+\sum_{j\in\{1,2,\ldots,\frac{m-2^l-2}{2}\}\setminus U}
	2^{j}\]
	and
	\[-ak_1\pmod{2^m-1}=(2^{2^l}-1)+(2^{\frac{m+2^l}{2}}+2^{2^l})\sum_{j\in\{1,2,\ldots,\frac{m-2^l-2}{2}\}\setminus U}
	2^{j}\]
	Therefore,  we have
	\begin{align}\min_{s}\{-2^{s}ak_1\pmod{2^m-1}\}
		\geq 2^{2^l}-1,\label{01}\end{align}
with equality if and only if $U=\{1,2,\ldots,\frac{m-2^l-2}{2}\}$, i.e., $ k_1=2^{\frac{m-2^l}{2}}-1$. This means
\[-ak_1\pmod{2^m-1}\equiv2^{2^l}-1\in\Gamma_1\ \text{if}\ h=2^l,\ \text{and}\ \notin\Gamma_1 \ \text{if}\ 0<h< 2^l.\]Note that $\frac{m}{2}>2^{l}$, by Lemma~\ref{1}, we have $v_{2^{2^l}-1}=m-2^l\equiv0\pmod2$. Then \[\kappa_{2^{2^l}-1}^{2^l}=1,\  \text{and}\ u_{2^{2^l}-1}=v_{2^{2^l}-1}+\kappa_{2^{2^l}-1}^{2^l}\pmod2\equiv1 \ \text{if}\ h=2^l,\]and $u_{2^{2^l}-1}=v_{2^{2^l}-1}\pmod2\equiv0 $ if $0<h< 2^l$.
Therefore, \[C_{2^{2^l}-1}\in D_{(0,m)}\ \text{if}\ 0<h< 2^l,\ \text{and}\  C_{2^{2^l}-1}\notin D_{(0,m)}\ \text{if}\ h=2^l.\]Note that \[ak_1\pmod{2^m-1}\equiv2^{m-2^l}-1\notin \Gamma_1,\ \text{and}\ u_{2^{m-2^l}-1}=2^l\pmod2\equiv0.\] Summarizing the discussions above yields the
desired conclusion.
\end{proof}

\subsection{{Lower bounds} on the minimum distance} In this subsection, we shall divide into two cases to discuss the {lower bounds on $d(\C_{\cD,0})$ and $d(\C_{\cD,1})$}.
\vskip 0.2cm

\noindent
\textbf{A. The case that $m$ is even}
\begin{theorem}\label{dz1}
	Let $m=2^le$, $l\geq1$ and $2\nmid e$, let $\C_{\cD,1}$ be defined by ~\eqref{degm}.  Then the code $\C_{\cD,1}$ has parameters $[2^m-1,2^{m-1}-2 ,d(\C_{\cD,1})]$, where \[d(\C_{\cD,1})\geq \begin{cases}
		2^{\frac{m}{2}}+4, & \mbox{if } l\geq2,m\geq8, 0<h\leq \frac{m-4}{2}, 2\mid h \\
		2^{\frac{m-2}{2}}+4, & \mbox{if } l=1, m\geq10,  4\leq h\leq \frac{m-6}{2}, 2\mid h\ \text{or}\ m=6.
	\end{cases}\]or
\[d(\C_{\cD,1})\geq \begin{cases}
	2^{\frac{m}{2}}+2, & \mbox{if } l\geq2, m\geq8, 0<h\leq \frac{m-4}{2}, 2\nmid h\ \text{or}\ m=4 \\
	2^{\frac{m-2}{2}}+2, & \mbox{if } l=1, m\geq10, 4<h\leq \frac{m-6}{2}, 2\nmid h\ \text{or}\ 1\leq h\leq3\ \text{or}\ m=6.
\end{cases}\]

If $e\geq3$ and $0< h\leq 2^{l}$, then the code $\C_{\cD,0}$ has parameters $[2^m-1,2^{m-1},d(\C_{\cD,0})]$, where
\[d(\C_{\cD,0})\geq \begin{cases}
	2^{\frac{m-2^l+2}{2}}+2, & \mbox{if } h\neq2^l \\
	2^{\frac{m-2^l+2}{2}}, & \mbox{if } h=2^l.
\end{cases}\]
In particular, if $l=1$, then
\begin{equation}\label{b1}d(\C_{\cD,0})\geq2^{\frac{m}{2}}+2\ \text{if}\ h=1,\  \text{and}\ \geq2^{\frac{m}{2}}\ \text{if}\ h=2.\end{equation}
\end{theorem}
\begin{proof}Note that in Lemmas~\ref{lemdz10}-\ref{lemdz12}, $\gcd(a,2^m-1)=1$, so there  exists an integer $a'$ such that $aa'\equiv1\pmod{2^m-1}$. Take $\gamma=\alpha^{a'}$. Then $\gamma$ is also an $n$-th primitive root of unity.

	We first show that the lower bounds on $d(\C_{\cD,1})$. In the case that $l\geq2$ and $m\geq8$, by Eq.~\eqref{degm} and Lemma~\ref{lemdz10}, the defining set of $\C_{\cD,1}$ with respect to $\gamma$ contains the sets
	\[\{k:1\leq k\leq 2^{\frac{m-2}{2}}\}\cup\{k:2^m-2^{\frac{m-2}{2}}-1\leq k\leq 2^{m}-1\}\ \text{if}\ 2\nmid h,\] and
	\[\{k:1\leq k\leq 2^{\frac{m-2}{2}}+2\}\cup\{k:2^m-2^{\frac{m-2}{2}}-1\leq k\leq 2^{m}-1\}\ \text{if}\ 2\mid h.\]
Thus the lower bounds on $d(\C_{\cD,1})$ can be obtained from the BCH bound. The case of $l=1$ and $m\geq10$ can be similarly obtained.

For $m=4$, let $a=7$ and $0<h\leq2$. Then $a^{-1}\pmod{15}=13$ and the defining set of $\C_{\cD,1}$ with respect to $\alpha^{13}$ contains the sets \[\{0,1,2,13,14\}\ \text{if}\ h=1, \text{and}\ \{0,1,2,3,4\}\ \text{if}\ h=2.\]
Therefore, $d(\C_{\cD,1})\geq6$.

For $m=6$, let $a=5$ and $0<h\leq2$. Then $a^{-1}\pmod{63}=38$ and the defining set of $\C_{\cD,1}$ with respect to $\alpha^{38}$ contains the sets \[\{0,1,2,61,62\}\ \text{if}\ h=1\ \text{and}\ \{0,1,2,55,56,57,58,59,60,61,62\}\ \text{if}\ h=2.\]
Therefore, $d(\C_{\cD,1})\geq6$ if $h=1$, $d(\C_{\cD,1})\geq12$ if $h=2$.

We next give the lower bounds on $d(\C_{\cD,0})$. By Eq.~\eqref{degm} and Lemma~\ref{lemdz12}, we obtain that the defining set of $\C_{\cD,0}$ with respect to $\gamma$ contains the sets \[\{k:1\leq k\leq 2^{\frac{m-2^l}{2}}\}\cup\{k:2^m-2^{\frac{m-2^l}{2}}-1\leq k\leq 2^{m}-1\}\ \text{if}\ h\neq 2^l.\] and
\[\{k:1\leq k\leq 2^{\frac{m-2^l}{2}}\}\cup\{k:2^m-2^{\frac{m-2^l}{2}}+1\leq k\leq 2^{m}-1\}\ \text{if}\ h= 2^l.\]
	Hence the desired lower bounds on $d(\C_{\cD,0})$ follow from the BCH
	bound on the cyclic codes.

%For $m=4$, let $a=13$ and $0<h\leq2$. Then $a^{-1}\pmod{15}=7$. Then the defining set of $\C_{\cD,0}$ with respect to $\alpha^{7}$ contains the set $ \{0,13,14\}$ if $h=2$.
%Therefore, $d(\C_{\cD,0})\geq4$. The proof can be completed.

For $m=6$, let $a=38$ and $h=1$. Then $a^{-1}\pmod{63}=5$ and the defining set of $\C_{\cD,0}$ with respect to $\alpha^{5}$ contains the set \[\{0,1,2,3,4,59,60,61,62\}\ \text{if}\ h=1\ \text{and}\ \{0,1,2,3,4,61,62\}\ \text{if}\ h=2.\]
Therefore, $d(\C_{\cD,0})\geq10$ if $h=1$, and $d(\C_{\cD,0})\geq8$ if $h=2$. This meets the lower bound given by ~\eqref{b1}. The proof can be completed.
	\end{proof}

By using Magma, we compute the following examples in Table~\ref{table:t2} to check Theorem~\ref{dz1}, and form the following conjecture.
\begin{table}
	\caption{\textbf{The minimum distances of $\C_{\cD,0}$ and $\C_{\cD,1}$}}%标题
	\centering%把表居中
\label{table:t2}
	\begin{tabular}{|c|ccc|ccc|}%五个c代表该表一共五列，内容全部居中
		\hline%第一道横线
		\multirow{2}{*}{$m$} &  & \multirow{1}{*}{$\C_{\cD,1}$} &  & & \multirow{1}{*}{$\C_{\cD,0}$}& \\
 & \multirow{1}{*}{$\dim(\C_{\cD,1})$} & \multirow{1}{*}{$h$} & \multirow{1}{*}{$d(\C_{\cD,1})$} & \multirow{1}{*}{$\dim(\C_{\cD,0})$}& \multirow{1}{*}{$h$}& \multirow{1}{*}{$d(\C_{\cD,0})$} \\
		\hline%第二道横线
		\multirow{2}{*}{4}&\multirow{2}{*}{6}&\multirow{1}{*}{1}&\multirow{1}{*}{6}&\multirow{2}{*}{8}&\multirow{2}{*}{1,2} &\multirow{2}{*}{4}  \\
		&                  & \multirow{1}{*}{2}                 &\multirow{1}{*}{6}& & &  \\
		\hline%第三道横线
		\multirow{2}{*}{6}&\multirow{2}{*}{30}&\multirow{1}{*}{1}&\multirow{1}{*}{6}&\multirow{2}{*}{32}& \multirow{1}{*}{1} & \multirow{1}{*}{10}  \\
		&                  & \multirow{1}{*}{2}  &\multirow{1}{*}{12}               & &\multirow{1}{*}{2} &\multirow{1}{*}{8}  \\
		\hline%第四道横线
		\multirow{2}{*}{8}&\multirow{2}{*}{126}&\multirow{1}{*}{1}&\multirow{1}{*}{18}&\multirow{2}{*}{128}&\multirow{2}{*}{1,2}&\multirow{2}{*}{22}  \\
		&                  &   \multirow{1}{*}{2}               &\multirow{1}{*}{20}& & & \\
		\hline%第五道横线
		\multirow{2}{*}{10}&\multirow{2}{*}{510}&\multirow{1}{*}{1}&\multirow{1}{*}{18}&\multirow{2}{*}{512}&\multirow{1}{*}{1}&\multirow{1}{*}{34}  \\
		&                  &\multirow{1}{*}{2}                  &\multirow{1}{*}{28}&&\multirow{1}{*}{2}&\multirow{1}{*}{32}  \\
		\hline%第六道横线
		\multirow{2}{*}{12}&\multirow{2}{*}{2046}&\multirow{1}{*}{1,3}&\multirow{1}{*}{66}&\multirow{2}{*}{2048}&\multirow{1}{*}{1,2,3}&\multirow{1}{*}{34}  \\
		&                  & \multirow{1}{*}{2,4}                 &\multirow{1}{*}{68}&& \multirow{1}{*}{4} &\multirow{1}{*}{32}  \\
		\hline%第七道横线
		\multirow{2}{*}{14}&\multirow{2}{*}{8190}&\multirow{1}{*}{1,2,3}&\multirow{1}{*}{66}&\multirow{2}{*}{8192}&\multirow{1}{*}{1}&\multirow{1}{*}{130}  \\
		&                  &  \multirow{1}{*}{4}                &\multirow{1}{*}{68}& &\multirow{1}{*}{2}&\multirow{1}{*}{128}  \\
		\hline%第八道横线
		\multirow{2}{*}{16}&\multirow{2}{*}{32766}&\multirow{1}{*}{1,3,5}&\multirow{1}{*}{258}&\multirow{2}{*}{32768}&\multirow{2}{*}{$1\leq h\leq 16$}&\multirow{2}{*}{146}  \\
		&                  & \multirow{1}{*}{2,4,6}                 &\multirow{1}{*}{260}&& & \\
		\hline%第九道横线
\multirow{2}{*}{18}&\multirow{2}{*}{131070}&\multirow{1}{*}{1,2,3,5}&\multirow{1}{*}{258}&\multirow{2}{*}{131072}&\multirow{1}{*}{1}&\multirow{1}{*}{514}  \\
		&                  &\multirow{1}{*}{4,6}                  &\multirow{1}{*}{260}&&\multirow{1}{*}{2}&\multirow{1}{*}{512}  \\
\hline%第九道横线
\multirow{2}{*}{20}&\multirow{2}{*}{524286}&\multirow{1}{*}{1,3,5,7}&\multirow{1}{*}{1026}&\multirow{2}{*}{524288}&\multirow{1}{*}{1,2,3}&\multirow{1}{*}{514}  \\
		&                  &\multirow{1}{*}{2,4,6,8}                  &\multirow{1}{*}{1028}&&\multirow{1}{*}{4}&\multirow{1}{*}{512}  \\
		\hline%第十道横线
\multirow{2}{*}{22}&\multirow{2}{*}{2097150}&\multirow{1}{*}{1,2,3,5,7}&\multirow{1}{*}{1026}&\multirow{2}{*}{2097152}&\multirow{1}{*}{1}&\multirow{1}{*}{2050}  \\
		&                  &\multirow{1}{*}{4,6,8}                  &\multirow{1}{*}{1028}&&\multirow{1}{*}{2}  &\multirow{1}{*}{2048}  \\
		\hline%第十道横线
\multirow{2}{*}{24}&\multirow{2}{*}{8388606}&\multirow{1}{*}{1,3,5,7}&\multirow{1}{*}{4098}&\multirow{2}{*}{8388608}&\multirow{1}{*}{1,2,3,4,5,6,7}&\multirow{1}{*}{514}  \\
		&                  &\multirow{1}{*}{2,4,6,8}                  &\multirow{1}{*}{4100}&&\multirow{1}{*}{8}  &\multirow{1}{*}{512}  \\
		\hline%第十道横线
\multirow{2}{*}{26}&\multirow{2}{*}{33554430}&\multirow{1}{*}{1,2,3,5,7,9}&\multirow{1}{*}{4098}&\multirow{2}{*}{33554432}&\multirow{1}{*}{1}&\multirow{1}{*}{8194}  \\
		&                  &\multirow{1}{*}{4,6,8,10}                  &\multirow{1}{*}{4100}&&\multirow{1}{*}{2}  &\multirow{1}{*}{8192}  \\
		\hline%第十道横线
	\end{tabular}
\end{table}

\begin{conj}
 If $l\geq2$ and $e>1$, then $d(\C_{\cD,1})=2^{\frac{m}{2}}+4$ for even $h$, and $d(\C_{\cD,1})=2^{\frac{m}{2}}+2$ for odd $h$. If $l=1$, then $d(\C_{\cD,0})=2^{\frac{m}{2}}+2$ for $h=1$, and $d(\C_{\cD,0})=2^{\frac{m}{2}}$ for $h=2$. %Furthermore, if $v_2(m)=1$ and $m+2$ is not $2$-power then $d(\C_{\cD,m+2,0})=2^{\frac{m}{2}}+2$ for $0<h<2^l$, and $d(\C_{\cD,m+2,0})=2^{\frac{m}{2}}$ for $h=2^l$.
\end{conj}

From Table~\ref{table:t2}, we can see that the code $\C_{\cD,0}$ may not have a good minimum distance when $m$ is $2$-power. For example, when $m=16$ and $1\leq h\leq 16$, {$d(\C_{\cD,0})=146$}, which can not satisfy $d(\C_{\cD,0})^2\geq n$. So we are not interested in the case that $m$ is $2$-power.
\vskip 0.2cm

\noindent
{\bf B. The case that $m$ is odd}
\begin{theorem}\label{3}
  Let $m\equiv1\pmod4\geq5$ be odd and $0<h\leq\frac{m-3}{2}$. Then $\C_{\cD,1}$ has parameters $[2^m-1,2^{m-1}-1,{d(\C_{\cD,1})}]$, where
  \[d(\C_{\cD,1})\geq 2^{\frac{m-1}{2}}+4.\]
The code $\C_{\cD,0}$ has parameters $[2^m-1,2^{m-1},d(\C_{\cD,0})]$, where
\[d(\C_{\cD,0})\geq \begin{cases}
	2^{\frac{m-1}{2}}+1, & \mbox{if } 2\mid h, \\
	2^{\frac{m-1}{2}}+3, & \mbox{if } 2\nmid h.
\end{cases}\]
\end{theorem}
\begin{proof}
  By Lemma~\ref{lem3}, we have $\gcd(a,2^m-1)=1$, so there exists an integer $a'$ such that $aa'\equiv1\pmod{2^m-1}$. Take $\gamma=\alpha^{a'}$. Then $\gamma$ is also an $n$-th primitive root of unity.
The defining set of $D_{(1,m)}$ with respect to $\gamma$ contains the set \[\{k:0\leq k\leq 2^{\frac{m-1}{2}}+2\}.\]

 The defining set of $\C_{\cD,0}$ with respect to $\gamma$ contains the sets  \[\{k:2^m-2^{\frac{m-1}{2}}-3\leq k\leq 2^m-2\}\ \text{if}\ 2\nmid h\ \text{and}\ \{k:2^m-2^{\frac{m-1}{2}}-1\leq k\leq 2^m-2\}\ \text{if}\ 2\mid h.\]
  Hence the desired lower bounds on {the minimum distances of $\C_{\cD,1}$ and $\C_{\cD,0}$} follow from the BCH bound.
\end{proof}
\begin{example}\label{e5}
	Let $m=5$ and $h=1$. Then $\C_{\cD,1}$ has  parameters $[31,15,8]$ and  generator polynomial $x^{16}+x^{15}+x^{14}+x^{11}+x^{10}+x^9+x^8+x^7+x^6+x^3+x^2+1$. The code $\C_{\cD,0}$ has parameters $[31,16,7]$ and generator polynomial $x^{15} + x^{14} + x^{12} + x^{11} + x^{10} + x^8 + x^6 + x^4 + x^3 + x^2 + 1$. The lower bounds of $d(\C_{\cD,1})$ and $d(\C_{\cD,0})$ given by Theorem~\ref{1} are achieved. Furthermore, the code $\C_{\cD,1}$ is optimal according to the Database.\end{example}

Let $m=9$ and $h=2$,  then the codes $\C_{\cD,1}$ and $\C_{\cD,0}$ respectively have parameters $[511,255,20]$ and $[511,256,20]$, and the coset leaders which contains their defining sets are respectively \begin{align*}D_{(1,m)}\cap\Gamma=\{&1,5,9,15,17,23,27,29,39,43,45,51,53,57,63,75,77,83,85,95,111,\\
&119,123,125,175,183,187,219,255\}\end{align*}and
\begin{align*}D_{(0,m)}\cap\Gamma=\{&3,7,11,13,19,21,25,31,35,37,41,47,55,59,61,73,79,87,91,93,103,\\
&107,109,117,127,171,191,239,223\}.
\end{align*}
However, in \cite{Tang1}, the codes $\C_{0,m}$ and $\C_{1,m}$ have parameters $[511,256,19]$, and the coset leaders which contains their defining sets are respectively
\[(D_{(1,m)}\cap\Gamma-\{1\})\cup\{3\},\ \text{and}\ (D_{(0,m)}\cap\Gamma-\{3\})\cup\{1\}.\]

 Therefore, the defining sets $D_{(1,m)}$ and $D_{(1,m)}$ are different from those of $\C_{1,m}$ and $\C_{0,m}$.
\begin{theorem}
	Let $m\equiv3\pmod4\geq7$ and $0<h\leq\frac{m-3}{2}$. Then $\C_{\cD,1}$ has parameters $[2^m-1,2^{m-1}-1,d(\C_{\cD,1})]$, where
	\[d(\C_{\cD,1})\geq \begin{cases}
		2^{\frac{m-1}{2}}+4, & \mbox{if } 2\mid h, \\
		2^{\frac{m-1}{2}}+2, & \mbox{if } 2\nmid h.
	\end{cases}\]
The code $\C_{\cD,0}$ has parameters $[2^m-1,2^{m-1},d(\C_{\cD,0})]$, where
\[d(\C_{\cD,0})\geq 
	2^{\frac{m-1}{2}}+1.\]
Moreover, if $m=3$, then the codes $\C_{\cD,1}$ and $\C_{\cD,0}$ has parameters $[7,3,4]$ and $[7,4,3]$ respectively.
\end{theorem}
\begin{proof}
	By Lemma~\ref{lem30}, we have $\gcd(a,2^m-1)=1$, so there exists an integer $a'$ such that $aa'\equiv1\pmod{2^m-1}$. Take $\gamma=\alpha^{a'}$. Then $\gamma$ is also an $n$-th primitive root of unity.
	In addition, the defining set of $\C_{\cD,1}$ with respect to $\gamma$ contains the sets \[\{k:0\leq k\leq 2^{\frac{m-1}{2}}+2\}\ \text{if}\ 2\mid h,\ \text{and}\ \{k:0\leq k\leq 2^{\frac{m-1}{2}}\}\ \text{if}\ 2\nmid h .\]
	
	The defining set of $\C_{\cD,0}$ with respect to $\gamma$ contains the sets  \[\{k:2^m-2^{\frac{m-1}{2}}-3\leq k\leq 2^m-2\}\ \text{if}\ 2\mid h,\ \text{and}\ \{k:2^m-2^{\frac{m-1}{2}}-1\leq k\leq 2^m-2\}\ \text{if}\ 2\nmid h .\]
	
	For $m=3$, the defining sets of $\C_{\cD,1}$ and $\C_{\cD,0}$ with respect to $\alpha^3$ contain the sets $\{0,1,2\}$ and $\{5,6\}$ if $h=1$, and $\{0,5,6\}$ and $\{1,2\}$ if $h=2$.
	Hence the desired lower bounds on $d(\C_{\cD,0})$ and $d(\C_{\cD,1})$ follow from the BCH bound.
\end{proof}

\begin{example}
  Let $m=3$. Then $\C_{\cD,1}$ has  parameters $[7,3,4]$ and  generator polynomial $x^4+x^2+x+1$ if $h=1$, and $x^4 + x^3 + x^2 + 1$ if $h=2$. The $\C_{\cD,0}$ has  parameters $[7,4,3]$
  and generator polynomial $x^3 + x + 1$ if $h=1$, and $x^3 + x^2 + 1$ if $h=2$. The codes  $\C_{\cD,1}$  and $\C_{\cD,0}$ are optimal according to the Database.
\end{example}
\begin{example}\label{ex:7}
  Let $m=7$. If $h=1$, then $\C_{\cD,1}$ has  parameters $[127,63,20]$ and  generator polynomial $x^{64} + x^{62} + x^{59} + x^{58} + x^{54} + x^{52} + x^{51} + x^{50}
  + x^{49} + x^{47} + x^{46} +x^{45} + x^{44} + x^{43}+ x^{41} + x^{40} + x^{39}
  + x^{38} + x^{37} + x^{36} + x^{33}+ x^{32} +x^{31} + x^{30} + x^{29} + x^{28}
  + x^{27} + x^{26} + x^{24} + x^{23} + x^{19} + x^{18} + x^{17} +
  x^{16}+ x^{15}+ x^{14} + x^{10} + x^8 + x^6 + x^4 + x^3 + x^2 + x + 1$, and $\C_{\cD,0}$ has  parameters $[127,64,19]$ and generator polynomial $x^{63} + x^{61} + x^{59} + x^{58} + x^{55} + x^{54} + x^{49} + x^{47} + x^{45} + x^{40} + x^{37} +
  x^{35} + x^{33} + x^{31} + x^{27} + x^{25 }+ x^{23 }+ x^{20} + x^{18} + x^{16 }+ x^{15} + x^{13} +
  x^{11 }+ x^{10} + x^5 + x + 1$.  If $h=2$,
  then the code $\C_{\cD,1}$ has parameters $[127,63,20]$ and generator polynomial $x^{64} + x^{63} + x^{62} + x^{61} + x^{60} + x^{58} + x^{56} + x^{54} + x^{50} + x^{49} + x^{48} +x^{47} + x^{46} + x^{45} + x^{41} + x^{40} + x^{38} + x^{37} + x^{36} + x^{35} + x^{34} + x^{33} +x^{32} + x^{31} + x^{28} + x^{27} + x^{26} + x^{25} + x^{24} + x^{23} + x^{21} + x^{20} + x^{19} +x^{18} + x^{17} + x^{15} + x^{14} + x^{13} + x^{12} + x^{10} + x^6 + x^5 + x^2 + 1$, and $\C_{\cD,0}$ has parameters $[127,64,19]$ and generator polynomial $ x^{63} + x^{59} + x^{58} + x^{56} + x^{55} + x^{54} + x^{53} + x^{52} + x^{51}
  + x^{47} + x^{45} + x^{44} +x^{43}+ x^{42}  + x^{41} + x^{40} + x^{37} + x^{36} + x^{34} + x^{33} +x^{31} + x^{30}  + x^{27} + x^{24} + x^{20}+ x^{15 }+ x^9 + x^8 + x^3 + x + 1$.
\end{example}

\section{Summary and Concluding Remarks}
In this paper, we employed two classes of sequences to construct four families of binary cyclic codes $\C_{\cS,i}$ and $\C_{\cD,i}$, $i=0,1$. When $m$ is even, there exist some codes such that their minimum distances $d$ satisfy $d^2\geq n$. Such as
\begin{itemize}
	\item the code $\C_{\cS,1}$ has length $2^m-1$, dimension $2^{m-1}-2$ and minimum distance $d(\C_{\cS,1})\geq 2^{\frac{m}{2}}+2$ if $l\geq2$.
	\item  the code $\C_{\cS,0}$ has length $2^m-1$, dimension $2^{m-1}$ and minimum distance $d(\C_{\cS,0})\geq 2^{\frac{m}{2}}+2$ if $l=1$ and $e\geq3$.
	\item the code $\C_{\cD,1}$ has length $2^m-1$, dimension $2^{m-1}-2$ and minimum distance $d(\C_{\cD,1})\geq 2^{\frac{m}{2}}+4$ if  $l\geq2$, $m\geq8$, $0<h\leq \frac{m-4}{2}$ and $2\mid h$, $\geq 2^{\frac{m}{2}}+2$ if $l\geq2$, $m\geq8$ and $2\nmid h$ or $m=4$.
	\item the code $\C_{\cD,0}$ has length $2^m-1$, dimension $2^{m-1}$ and minimum distance $d(\C_{\cD,0})\geq2^{\frac{m}{2}}+2$ if $l=h=1$ and $e\geq3$.
\end{itemize}

When $m$ is odd, there exist some codes with good minimum distances. Such as
\begin{itemize}
  \item When $m=3$. The codes $\C_{\cD,1}$ and $\C_{\cD,0}$ have parameters $[7,3,4]$ and $[7,4,3]$ respectively, they are the best cyclic codes according to the Database.
  \item When $m=5$. If $h=1$, then the codes $\C_{\cD,1}$ and $\C_{\cD,0}$ have parameters $[31,15,8]$ and $[31,16,7]$, and their defining sets are respectively
      \[C_0\cup C_3\cup C_5\cup C_{15}\ \text{and}\ C_1\cup C_7\cup C_{11}. \]
      In \cite[Example~19]{Sunc}, Sun et al.\ also gave the code $\C(5)$ with parameters $[31,16,7]$ and defining set $C_1\cup C_3\cup C_5$. By replacing $\alpha$ by $\alpha^{25}$, the set $C_1\cup C_3\cup C_5$ can be transformed into $C_1\cup C_7\cup C_{11}$. Thus, the code $\C_{\cD,0}$ is equivalent to $\C(5)$. However,
      if $h=2$,  the code $\C_{\cD,0}$ {has} parameters $[31,16,6]$ and the defining set $C_3\cup C_7\cup C_{11}$. It can be checked that the codes $\C_{\cD,0}$ and $\C(5)$ are not equivalent.
  \item When $m=7$. If $1\leq h\leq 2$, the codes $\C_{\cD,1}$ and $\C_{\cD,0}$ have parameters $[127,63,20]$ and $[127,64,19]$ respectively. In addition, if $h=1$, then $\C_{\cD,0}=\C_{1,7}$, and the defining set is
      \[C_1\cup C_7\cup C_{11}\cup C_{13}\cup C_{19}\cup C_{21}\cup C_{31}\cup C_{47}\cup C_{55}.\]  If $h=2$, then $\C_{\cD,0}=\C_{0,7}$, and the defining set is \[C_3\cup C_5\cup C_9\cup C_{15}\cup C_{23}\cup C_{27}\cup C_{29}\cup C_{43}\cup C_{63}.\]If $h=4$, the code  $\C_{\cD,0}$ have parameters $[127,64,20]$, and its defining set is
      \[C_3\cup C_7\cup C_9\cup C_{15}\cup C_{19}\cup C_{21}\cup C_{31}\cup C_{47}\cup C_{55},\]it can be checked that the code  $\C_{\cD,0}$ is not equivalent to $\C_{0,7}$ and $\C_{1,7}$.
\end{itemize}

From the above examples, we can see that the code $\C_{\cD,i}$ may have excellent minimum distance if we choose some other $h$. Therefore, we invite readers to give a tight lower bounds on the minimum distance of $\C_{\cD,i}$, $i=1,2$.

 Finally, we believe that it would be very interesting to construct more binary cyclic codes of length $n$ and dimension near $(n\pm1)/2$ whose minimum distances satisfy $d^2\geq n$.

\end{document}